\documentclass[12pt]{article}%
\usepackage{amsmath}%
\usepackage{amssymb}%
\usepackage{xcolor}%
\usepackage{graphicx}%
\usepackage[cm]{fullpage}%

\usepackage[amsmath,thmmarks]{ntheorem}%
\theoremseparator{.}%
\usepackage{titlesec}%
\titlelabel{\thetitle. }%

\usepackage{xcolor}%
\usepackage{mleftright}%
\usepackage{xspace}%
\usepackage{caption}%
\usepackage{scalerel}%
\usepackage{euscript}%

\usepackage{hyperref}%
\usepackage{longtable}%

\usepackage{hyperref}%
\hypersetup{%
   breaklinks,%
   ocgcolorlinks, colorlinks=true,%
   urlcolor=[rgb]{0.25,0.0,0.0},%
   linkcolor=[rgb]{0.5,0.0,0.0},%
   citecolor=[rgb]{0,0.2,0.445},%
   filecolor=[rgb]{0,0,0.4}, anchorcolor=[rgb]={0.0,0.1,0.2}%
}

\theoremseparator{.}%

\theoremstyle{plain}%
\newtheorem{theorem}{Theorem}[section]

\newtheorem{lemma}[theorem]{Lemma}

\newtheorem{corollary}[theorem]{Corollary}

\theoremstyle{plain}%
\theoremheaderfont{\sf} \theorembodyfont{\upshape}%
\newtheorem*{remark:unnumbered}[theorem]{Remark}%

\newtheorem{defn}[theorem]{Definition}

\newtheorem{problem}[theorem]{Problem}

\newcommand{\myqedsymbol}{\rule{2mm}{2mm}}

\theoremheaderfont{\em}%
\theorembodyfont{\upshape}%
\theoremstyle{nonumberplain}%
\theoremseparator{}%
\theoremsymbol{\myqedsymbol}%
\newtheorem{proof}{Proof:}%

\newcommand{\atgen}{\symbol{'100}}
\newcommand{\SarielThanks}[1]{\thanks{Department of Computer Science;
      University of Illinois; 201 N. Goodwin Avenue; Urbana, IL,
      61801, USA; {\tt sariel\atgen{}illinois.edu}; {\tt
         \url{http://sarielhp.org/}.} #1}}

\newcommand{\HLink}[2]{\hyperref[#2]{#1~\ref*{#2}}}
\newcommand{\HLinkSuffix}[3]{\hyperref[#2]{#1\ref*{#2}{#3}}}

\newcommand{\figlab}[1]{\label{fig:#1}}
\newcommand{\figref}[1]{\HLink{Figure}{fig:#1}}

\newcommand{\thmlab}[1]{{\label{theo:#1}}}
\newcommand{\thmref}[1]{\HLink{Theorem}{theo:#1}}

\newcommand{\corlab}[1]{\label{cor:#1}}
\newcommand{\corref}[1]{\HLink{Corollary}{cor:#1}}%

\newcommand{\defrefY}[2]{\hyperref[def:#2]{#1}}

\newcommand{\lemlab}[1]{\label{lemma:#1}}
\newcommand{\lemref}[1]{\HLink{Lemma}{lemma:#1}}%

\newcommand{\problab}[1]{\label{prob:#1}}%
\newcommand{\probref}[1]{\HLink{Problem}{prob:#1}}%

\newcommand{\apndlab}[1]{\label{apnd:#1}}
\newcommand{\apndref}[1]{\HLink{Appendix}{apnd:#1}}

\newcommand{\seclab}[1]{\label{sec:#1}}
\newcommand{\secref}[1]{\HLink{Section}{sec:#1}}

\providecommand{\eqlab}[1]{}%
\renewcommand{\eqlab}[1]{\label{equation:#1}}

\newcommand{\Set}[2]{\left\{ #1 \;\middle\vert\; #2 \right\}}
\newcommand{\pth}[2][\!]{\mleft({#2}\mright)}%

\newcommand{\ceil}[1]{\left\lceil {#1} \right\rceil}

\newcommand{\brc}[1]{\left\{ {#1} \right\}}
\newcommand{\cardin}[1]{\left| {#1} \right|}%

\renewcommand{\th}{th\xspace} %

\usepackage[inline]{enumitem}

\newlist{compactenume}{enumerate}{5}%
\setlist[compactenume]{topsep=0pt,itemsep=-1ex,partopsep=1ex,parsep=1ex,%
   label={\bf\arabic*.}}%

\newlist{compactenumA}{enumerate}{5}%
\setlist[compactenumA]{topsep=0pt,itemsep=-1ex,partopsep=1ex,parsep=1ex,%
   label=(\Alph*)}%

\newlist{compactenuma}{enumerate}{5}%
\setlist[compactenuma]{topsep=0pt,itemsep=-1ex,partopsep=1ex,parsep=1ex,%
   label=(\alph*)}%

\newlist{compactenumI}{enumerate}{5}%
\setlist[compactenumI]{topsep=0pt,itemsep=-1ex,partopsep=1ex,parsep=1ex,%
   label=(\Roman*)}%

\newlist{compactenumi}{enumerate}{5}%
\setlist[compactenumi]{topsep=0pt,itemsep=-1ex,partopsep=1ex,parsep=1ex,%
   label=(\roman*)}%

\newlist{compactenumi*}{enumerate*}{5}%
\setlist[compactenumi*]{topsep=0pt,itemsep=-1ex,partopsep=1ex,parsep=1ex,%
   label=(\roman*)}%

\newlist{compactitem}{itemize}{5}%
\setlist[compactitem]{topsep=0pt,itemsep=-1ex,partopsep=1ex,parsep=1ex,%
   label={{\textbullet}}}%

\newcommand{\hrefb}[3][black]{\href{#2}{\color{#1}{#3}}}%

\renewcommand{\Re}{\mathbb{R}}%
\newcommand{\pB}{{q}}%
\newcommand{\eps}{{\varepsilon}}%
\newcommand{\Line}{{\ell}}%

\newcommand{\Spl}{{S}}%

\providecommand{\si}[1]{#1}
\newcommand{\tldO}{\scalerel*{\widetilde{O}}{j^2}}%
\newcommand{\areaX}[1]{\mathrm{area}\pth{#1}}

\newcommand{\Alg}{\texttt{alg}\xspace}%
\newcommand{\Talg}{T_{\text{\Alg}}}

 \newcommand{\rect}{{R}}%

\newcommand{\etal}{\textit{et~al.}\xspace}

\newcommand{\optY}[2]{\mathrm{opt}\pth{#1, #2}}%
\newcommand{\Tconvol}{T_{\mbox{\scriptsize\rm convol}}}
\newcommand{\RectAll}{{\EuScript{R}}}%
\newcommand{\rmin}{\rect_{\mathrm{min}}}%
\newcommand{\refX}[1]{{\updownarrow}{#1}}%
\newcommand{\foldX}[1]{\cardin{#1}_y}

\newcommand{\SaveContent}[2]{%
   \expandafter\newcommand{#1}{#2}%
}

\newcommand{\RestatementOf}[2]{
   \noindent%
   \textbf{Restatement of #1.}
   {\em #2{}}%
} \newcommand{\Jorgensen}{J{\o}rgensen\xspace}%
\newcommand{\Family}{{\mathcal{F}}}%
\newcommand{\optMinArea}{\alpha^*}

\definecolor{blue25emph}{rgb}{0, 0, 11} \providecommand{\emphic}[2]{%
   \textcolor{blue25emph}{%
      \textbf{\emph{#1}}}%
   \index{#2}}

\providecommand{\emphi}[1]{\emphic{#1}{#1}}

\newcommand{\Puatracscu}{P\u{a}tra\c{s}cu\xspace}

\newcommand{\RR}{{\cal R}}%
\newcommand{\II}{{\cal I}}

\numberwithin{figure}{section}%
\numberwithin{table}{section}%
\numberwithin{equation}{section}%

\newcommand{\IGNORE}[1]{}

\newcommand{\areaVal}{{\alpha}}%
\newcommand{\QQ}{\tau}%
\providecommand{\Matousek}{Matou{\v s}ek\xspace}

\title{Smallest $k$-Enclosing Rectangle Revisited}

\author{%
   Timothy M. Chan\thanks{Department of Computer Science; University
      of Illinois; 201 N. Goodwin Avenue; Urbana, IL, 61801, USA; {\tt
         tmc\atgen{}illinois.edu}; Work was partially supported by an
      NSF AF award CCF-1814026.}%
   \and%
   Sariel Har-Peled\SarielThanks{Work on this paper was partially
      supported by a NSF AF awards CCF-1421231,
      and %
      CCF-1217462.  %
   }}

\date{\today}

\begin{document}

\maketitle

\begin{abstract}
    Given a set of $n$ points in the plane, and a parameter $k$, we
    consider the problem of computing the minimum (perimeter or area)
    axis-aligned rectangle enclosing $k$ points. We present the
    first near quadratic time algorithm for this problem, improving
    over the previous near-$O(n^{5/2})$-time algorithm by Kaplan \etal
    \cite{krs-faprf-17}. We provide an almost matching conditional
    lower bound, under the assumption that $(\min,+)$-convolution
    cannot be solved in truly subquadratic time.  Furthermore, we
    present a new reduction (for either perimeter or area) that can
    make the time bound sensitive to $k$, giving near $O(n k) $
    time. We also present a near linear time $(1+\eps)$-approximation
    algorithm to the minimum area of the optimal rectangle containing
    $k$ points. In addition, we study related problems including the
    $3$-sided, arbitrarily oriented, weighted, and subset sum versions
    of the problem.
\end{abstract}

\section{Introduction}

Given a set $ P$ of $n$ points in the plane, and a parameter $k$,
consider the problem of computing the smallest area/perimeter
axis-aligned rectangle that contains $k$ points of $ P$. (Unless
stated otherwise, rectangles are axis-aligned by default.)  This
problem and its variants have a long history.  Eppstein and Erickson
\cite{ee-innfm-94} studied an exhaustive number of variants of this
problem for various shapes.

For the minimum perimeter variant, the first work on this problem
seems to be Aggarwal \etal \cite{aiks-fkpmd-91}, who showed a brute
force algorithm with running time $O(n^3)$. Recently, Kaplan \etal
\cite{krs-faprf-17} gave an algorithm with running time
$O(n^{5/2}\log^2n)$ that works for both minimum perimeter and area.

Several works derived algorithms with running time sensitive to~$k$,
the number of points in the shape.
Aggarwal \etal~\cite{aiks-fkpmd-91} showed an algorithm for the
minimum perimeter with running time $O(k^2 n \log n)$. This was
improved to $O(n \log n + k^2 n )$ by Eppstein and Erickson
\cite{ee-innfm-94} or alternatively by Datta \etal~\cite{Datta}.
Kaplan \etal's algorithm \cite{krs-faprf-17} for the $k$-insensitive
case, coupled with these previous techniques \cite{ee-innfm-94,Datta},
results in an $O( n\log n + nk^{3/2} \log^2 k)$ running time, which is
currently the state of the art.

\renewcommand{\tldO}{\widetilde{O}} Known techniques
\cite{ee-innfm-94,Datta} reduce the problem to solving $O(n/k)$
instances of size $O(k)$. These reductions work only for the perimeter
case, not the area case -- in particular, there are incorrect
attributions in the literature to results on the minimum area
rectangle -- see the introduction of d{}e~Berg \etal
\cite{bccek-cmpsa-16} for details.  \si{De}~Berg \etal described an
algorithm with running time $O(n\log^2n + nk^2 \log n )$ for minimum
area. Both d{}e~Berg \etal \cite{bccek-cmpsa-16} and Kaplan \etal
\cite{krs-faprf-17} left as an open question whether there is a
reduction from the minimum-area problem to about $\tldO(n/k)$
instances of size $O(k)$, where $\tldO$ hides\footnote{We reserve the
   right, in the future, to use the $\tldO$ to hide any other things
   we do not like.}  polynomial factors in $\log n$ and $1/\eps$. Such
a reduction would readily imply an improved algorithm.

\paragraph*{Our results.}
We revisit the above problems and provide significantly improved
algorithms:

\begin{compactenumA}[resume=results]
    \item \textsf{Exact smallest $k$-enclosing rectangle.} In
    \secref{krect} we describe an algorithm for the minimum
    $k$-enclosing rectangle (either area or perimeter) with running
    time $O(n^2 \log n)$ (see \thmref{krect}).  It is based on a new
    divide-and-conquer approach, which is arguably simpler than Kaplan
    \etal's algorithm. Known reductions mentioned above then lead to
    an $O(n\log n + nk \log k)$-time algorithm for computing the
    minimum perimeter rectangle.

    \smallskip%
    \item \textsf{$k$-sensitive running time for smallest area.}  In
    \secref{reduction:k:sensitive} we describe a reduction of the
    minimum-area problem to $O( \tfrac{n}{k} \log \tfrac{n}{k})$
    instances of size $O(k)$ (see \thmref{k:s:line}).  Our reduction
    uses \emph{shallow cutting} for 3-sided rectangular ranges
    \cite{jl-rsmt-11}
    (see \apndref{shallow:cutting}) and is conceptually simple.

    Plugging this the aforementioned new $O(n^2\log n)$-time algorithm
    leads to $O(nk \log \tfrac{n}{k} \log k )$-time algorithm for
    computing the minimum area $k$-enclosing rectangle (see
    \corref{krect:sensitive:area}).  Thus, our new result strictly
    improves upon both Kaplan \etal's and de Berg \etal's results for
    all $k$, from constant to $\Theta(n)$.
\end{compactenumA}
\medskip%
The smallest enclosing rectangle problem is amenable to sampling.
Kaplan \etal\ used samples in an approximation algorithm, with running
time $\tldO( n/ k )$, that computes a rectangle containing at least
$(1-\eps)k$ points of a prescribed perimeter, where $k$ is the
maximum number of points in any such rectangle. Similarly, using
relative approximations \cite{hs-rag-11}, d{}e Berg \etal
\cite{bccek-cmpsa-16} showed an algorithm that computes, in $\tldO(n)$
time, a rectangle containing $\geq (1-\eps)k$ points, where $k$ is
the maximum number of points in any rectangle of a prescribed area.
The ``dual'' problem, of approximating the minimum area rectangle
containing $k$ points seems harder, since sampling does not directly
apply to it.  \medskip%
\begin{compactenumA}[resume=results]
    \item \textsf{Approximating the area of the smallest $k$-enclosing
       rectangle.} In \secref{min:area:approx}, we present an
    approximation algorithm that computes, in $O( n\log n )$ expected
    time, a rectangle containing $k$ points of area
    $\leq (1+\eps)\optMinArea$, for a constant $\eps \in (0,1)$, where
    $\optMinArea$ is the smallest-area of such a rectangle (see
    \thmref{smallest:area:approx}).
\end{compactenumA}
\medskip%

\noindent%
We next present a flotilla of related results: \smallskip%

\begin{compactenumA}[resume=results]
    \item \textsf{$3$-sided smallest $k$-enclosing rectangle.}  In
    \secref{3:sided} we (slightly) speed up the exact algorithm for
    the $3$-sided rectangles case (i.e., rectangles that must have
    their bottom edge on the $x$-axis). The running time is
    $O\bigl(n^2/2^{\Omega(\sqrt{\log n})} \bigr)$, and is obtained
    using known results on the \emph{(min,+)-convolution}
    problem~\cite{Bremner,w-fapsp-14} (see \thmref{krect:3sided}).

    \smallskip%
    \item \textsf{Arbitrarily oriented smallest $k$-enclosing
       rectangle.}  In \secref{arb:oriented:excluding} we briefly
    consider the variant where the rectangle may not be axis-aligned.
    We show that this problem can be solved in
    $O(n^3\log n + n^3k/2^{\Omega(\sqrt{\log k})})$ time, slightly
    improving a previous result of $O(n^3k)$~\cite{dgn-skper-05} when
    $k$ is not too small.

    \smallskip%
    \item \textsf{Minimum-weight $k$-enclosing rectangle.}  In
    \secref{krect:minwt} we show how to extend our $O(n^2\log n)$-time
    algorithm to the related problem of finding a minimum-weight
    rectangle that contains $k$ points, for $n$ given weighted points
    in the plane (see \thmref{krect:minwt}).

    \smallskip%
    \item \textsf{Subset sum for $k$-enclosing rectangle.} %
    In \secref{subset:sum:rect}, we study the problem of finding a
    rectangle that contains $k$ points and has a prescribed weight $W$
    (or as close as one can get to it). The running time of the new
    algorithm is $O(n^{5/2}\log n )$ (see \thmref{subset:sum}).

    \smallskip%
    \item \textsf{Conditional lower bound.}  In \secref{lower:bounds},
    we prove that our near quadratic algorithm for exact minimum
    (perimeter or area) $k$-enclosing rectangle is near optimal up to
    an arbitrarily small polynomial factor, under a ``popular''
    conjecture that the (min,+)-convolution problem cannot be solved
    in truly subquadratic time~\cite{cmww-pempc-17}.
\end{compactenumA}

\section{Smallest $k$-enclosing rectangle}

\subsection{An exact near-quadratic algorithm}%
\seclab{krect}

Our $O(n^2\log n)$-time algorithm for minimum $k$-enclosing rectangles
is based on divide-and-conquer.  It has some similarity with an
$O(n^2)$-time divide-and-conquer algorithm by Barbay \etal
\cite{bcnp-mwpb-14} for a different problem (finding the
minimum-weight rectangle for $n$ weighted points in the plane, without
any $k$-enclosing constraint), but the new algorithm requires more
ingenuity.

We start with a semi-dynamic data structure for a 1D subproblem:

\begin{lemma}
    \lemlab{1d}%
    Given a set $ P$ of $n$ points in 1D with $q$ \emph{marked}
    points, and an integer $k$, we can maintain an $O(q^2)$-space
    data structure, with $O(n\log n + nq)$ preprocessing time, that
    supports the following operations:
    
    \begin{compactitem}
        \item report the shortest interval containing $k$ points of
        $ P$ in $O(q)$ time;
        \item delete a marked point in $O(q)$ time;
        \item unmark a marked point in $O(q)$ time.
    \end{compactitem}
\end{lemma}

\begin{proof}
    Sort the points $ P$, and let $ p_1, \ldots,  p_n$ he resulting
    order. Consider the (implicit) matrix $M =  P- P$. Formally, the
    entry $M_{ij}$ is $ p_j -  p_i$ (we are interested only in the
    top right part of this matrix) -- such an entry can be computed in
    $O(1)$ time directly from the sorted point set. The optimal
    quantity of interest is the minimum on the $k$\th diagonal; that
    is, $\alpha(M) = \min_{i} M_{i,i+k-1}$. When a marked point get
    deleted, this corresponds to deleting a row and a column of $M$ --
    the quantity of interest remains the minimum along the $k$\th
    diagonal. Such a deletion, as far as a specific entry of the top
    right of the matrix is concerned, either
    \begin{compactenumi*}
        \item removes it,
        \item keeps it in its place,
        \item shift it one diagonal down as its moves left, or
        \item keep it on the same diagonal as it shifts both up and
        left (see \figref{matrix}).
    \end{compactenumi*}

    In particular, any sequence of at most $q$ deletions of elements
    can shift an entry in the matrix at most $q$ diagonals down. This
    implies that we need to keep track only of the $k,\ldots, k+q$
    diagonals of this matrix. To do better, observe that if we track
    the elements of an original diagonal of interest, the deletions
    can fragment the diagonal into at most $O(q)$ groups, where each
    group still appear as contiguous run of the original diagonal.

    To this end, let a \emphi{fragment} of a diagonal be either
    \begin{compactenumi*}
        \item a singleton entry that appears in a row or column of a
        marked point, or
        \item a maximum contiguous portion of the diagonal which does
        not touch any singleton entries from (i).
    \end{compactenumi*}
    It is easy to verify that the $k$\th diagonal of the matrix at any
    given point in time is made out of a sequence of at most $3k$
    fragments, where each fragment is an original fragment of one of
    the diagonals in the range $k, \ldots, k+q$.

    \begin{figure}[h]
        \phantom{}\hfill%
        \includegraphics[page=1]{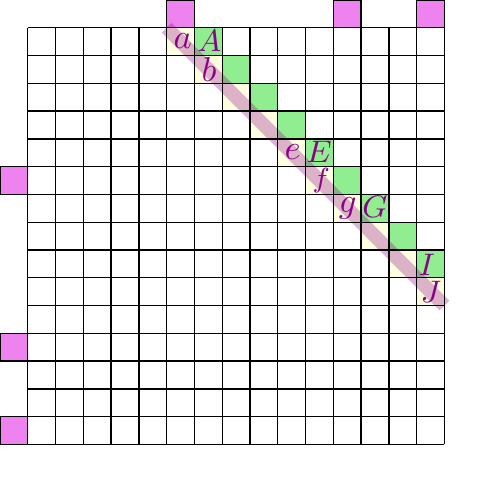}%
        \hfill%
        \includegraphics[page=2]{figs/matrix}%
        \hfill%
        \includegraphics[page=3]{figs/matrix}%
        \phantom{}\hfill\phantom{}%
        \caption{}%
        \figlab{matrix}
    \end{figure}

    As such, instead of storing all the elements of a fragment, we
    only maintain the minimum entry of the fragment (together with the
    information of what pairs of points it corresponds to). After this
    compression, a diagonal of interest can be represented as a linked
    list of $O(q)$ fragment summaries. In the preprocessing stage, the
    algorithm computes this representation for the $k$ to $k+q$
    diagonals (using this representation). This requires $O(q^2)$
    space, and $O(nq)$ time.

    A deletion of a marked point then corresponds to taking a
    contiguous block of linked fragments at the $i$\th list and moving
    it to list $i-1$, doing this surgery for $i=k,\ldots, k+q$. The
    blocks being moved start and end in singleton entries that
    correspond to the deleted point. We also need to remove these two
    singleton elements, and merge the two adjacent fragment summaries
    that are no longer separated by a singleton.  This surgery for all
    the $q+1$ lists of interest can be done in $O(q)$ time (we omit
    the tedious but straightforward details).

    A query corresponds to scanning the $k$\th diagonal and reporting
    the minimum value stored along it. An unmarking operation
    corresponds to merging two fragment summaries and the singleton
    separating them into a single fragment summary, and doing this for
    all the $q+1$ lists. Both operations clearly can be done in $O(q)$
    time.
\end{proof}

\begin{theorem}%
    \thmlab{krect}%
    Given a set $ P$ of $n$ points in the plane and an integer $k$,
    one can compute, in $O(n^2\log n)$ time, the
    smallest-area/perimeter axis-aligned rectangle enclosing $k$
    points.
\end{theorem}
\begin{proof}
    We do divide-and-conquer by $y$-coordinates.  Given a set $ P$ of
    $n$ points in the plane, and horizontal slabs $\sigma$ and $\tau$,
    each containing $q$ points of $ P$, we describe a recursive
    algorithm to find a smallest $k$-enclosing axis-aligned
    rectangle containing $ P$, under the restriction that the top
    edge is inside $\sigma$ and the bottom edge is inside $\tau$.  It
    is assumed that either $\sigma$ is completely above $\tau$, or
    $\sigma=\tau$.  It is also assumed that all points above $\sigma$
    or below $\tau$ have already been deleted from $ P$.  There can
    still be a large number of points in $P-(\sigma\cup\tau)$
    (recursion will lower $q$ but not necessarily $n$).  We will not
    explicitly store the points in $P-(\sigma\cup\tau)$, but rather
    ``summarize'' the points in an $O(q^2)$-space structure.  Namely,
    we assume that the $x$-coordinates of $ P$ are maintained in the
    1D data structure ${\cal S}$ of \lemref{1d}, where the marked
    points are the $O(q)$ points in $P\cap(\sigma\cup\tau)$.

    The algorithm proceeds as follows:
    \begin{compactenume}[start=0]
        \smallskip
        \item If $q=1$, then report the answer by querying ${\cal S}$
        in $O(1)$ time.  Else:

        \smallskip
        \item Divide $\sigma$ into two horizontal subslabs $\sigma_1$
        and $\sigma_2$, each containing $q/2$ points of $ P$.
        Likewise divide $\tau$ into $\tau_1$ and $\tau_2$.

        \smallskip
        \item For each $i,j\in\{1,2\}$, recursively solve the problem
        for the slabs $\sigma_i$ and $\tau_j$;\footnote{If
           $\sigma=\tau$, one of the four recursive calls is
           unnecessary.} to prepare for the recursive call, make a
        copy of ${\cal S}$, delete the (marked) points in
        $P\cap(\sigma\cup\tau)$ above $\sigma_i$ or below $\tau_j$,
        and unmark the remaining points in
        $P\cap(\sigma\cup\tau)-(\sigma_i\cup\tau_j)$, as shown in
        \figref{recursion}.  The time needed for these $O(q)$
        deletions and unmarkings, and for copying ${\cal S}$, is
        $O(q^2)$ (we emphasize that this bound is independent of $n$).
    \end{compactenume}

    \begin{figure}[h]
        \noindent%
        \hfill \includegraphics[page=1]{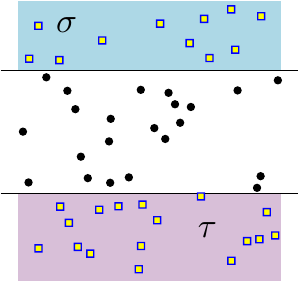} \hfill
        \includegraphics[page=2]{figs/div_dav} \hfill\phantom{}%
        
        \smallskip%

        \includegraphics[page=3]{figs/div_dav} \hfill%
        \includegraphics[page=4]{figs/div_dav} \hfill%
        \includegraphics[page=5]{figs/div_dav} \hfill%
        \includegraphics[page=6]{figs/div_dav}
        \caption{}
        \figlab{recursion}
    \end{figure}
    
    The running time satisfies the recurrence
    \begin{equation*}
        T(n,q) = 4\,T(n,q/2) + O(q^2),
    \end{equation*}
    with $T(n,1)=O(1)$, which gives $T(n,q)=O(q^2\log q)$.  Initially,
    $\sigma=\tau$ is the entire plane, with $q=n$; the data structure
    ${\cal S}$ can be preprocessed in $O(n^2)$ time.  Thus, the total
    running time is $O(n^2\log n)$.
\end{proof}

One can readily get an algorithm with $k$-sensitive running time for
the perimeter case, by reducing the problem into $O(n/k)$ instance
of size $O(k)$. This reduction is well known \cite{ee-innfm-94,Datta}
in this case -- approximate the smallest enclosing disk containing $k$
points in $O(n\log n)$ time, partition the plane into a grid with side
length proportional to the radius of this disk, and then solve the
problem for each cluster (i.e., $3\times 3$ group of grid cells) that
contains at least $k$ points of $ P$, using our above algorithm. We
thus get the following.

\begin{corollary}%
    \thmlab{krect:sensitive}%
    Given a set $ P$ of $n$ points in the plane and an integer $k$,
    one can compute, in $O(n\log n + nk \log k)$ time, the
    smallest-perimeter axis-aligned rectangle enclosing $k$ points
    of $ P$.
\end{corollary}

The $O(n\log n)$ term can be eliminated in the word RAM model, using a
randomized linear-time algorithm for approximate smallest
$k$-enclosing disk \cite{hm-facsk-05} (which requires integer division
and hashing).

A similar reduction for the minimum-area case is more challenging, and
was left as an open problem in previous work \cite{krs-faprf-17}. The
difficulty arises because the optimal-area rectangle may be long and
thin, with side length potentially much bigger than the radius of the
minimum $k$-enclosing disk.  Nonetheless, we show that such a
reduction is possible (with an extra logarithmic factor) in the next
subsection.

\subsection{$k$-sensitive running time for smallest area}
\seclab{reduction:k:sensitive}

Our starting point is a shallow cutting lemma for 3-sided
ranges~\cite{jl-rsmt-11}, which we describe in detail in
\apndref{shallow:cutting} for the sake of completeness.  (It can be
viewed as an orthogonal variant of \Matousek{}'s shallow cutting lemma
for halfspaces~\cite{Mat92}.)

\SaveContent{\LemmaShallowCutting}{%
   Given a set $ P$ of $n$ points in the plane, lying above a
   horizontal line $\Line$, and a parameter $k$, one can compute a
   family $\Family$ of at most $ 2\ceil{n/k}$ subsets of $ P$, each
   of size at most $6k$. The collection of sets can be computed in
   $O(n)$ time if the $x$-coordinates have been pre-sorted. For any
   axis-aligned rectangle $\rect$ with its bottom edge lying on
   $\Line$, that contains less than $k$ points of $ P$, we have
   $ P \cap \rect \subseteq  Q$ for some $ Q \in \Family$.  }

\begin{lemma}[\cite{jl-rsmt-11}]
    \lemlab{decomp:k}%
    \LemmaShallowCutting{}
\end{lemma}

\begin{defn}
    Let $\RectAll$ be the set of all axis-aligned rectangles in the
    plane. A \emphi{scoring function} is a function
    $f: \RectAll \rightarrow \Re$, with the following properties:
    
    \begin{compactenumA}
        \item Translation invariant: $\forall  p \in \Re^2$, we have
        $f( p + \rect) = f(\rect)$.

        \smallskip%
        \item Monotonicity: $\forall \rect, \rect' \in \RectAll$, such
        that $\rect \subseteq \rect'$, we have that
        $f(\rect)\leq f(\rect')$.
    \end{compactenumA}
\end{defn}

Functions that satisfy the above definition include area, perimeter,
diameter, and enclosing radius of a rectangle.

For a set $U \subseteq \Re^d$, let
$\refX{U} = \Set{(x,-y)}{(x,y) \in U}$ be the \emphi{reflection} of
$U$ through the $x$-axis.  Similarly, let
$\foldX{U} = \Set{\bigl.(x,|y|)}{(x,y) \in U}$ be the \emphi{folding}
of $U$ through the $x$-axis.

\begin{lemma}
    \lemlab{family:horizontal}%
    Given a set $ P$ of $n$ points in the plane, a parameter $k$, and
    a scoring function $f$, let $\rmin$ be the minimum score
    axis-aligned rectangle that contains $k$ points of $ P$, and
    intersects the $x$-axis.  Then, one can compute a family $\Family$
    of $O( n/k )$ subsets of $ P$, such that (i) each set of
    $\Family$ is of size $\leq 12k$, and (ii)
    $\rmin \cap  P \subseteq  Q$, for some $ Q \in \Family$.
\end{lemma}

\begin{proof}
    Let $ P' = \foldX{ P}$ be the ``folding'' of $ P$ over the
    $x$-axis, and let $\Family'$ be the cover of $ P'$ by sets of
    size $\leq 12k$, as computed by \lemref{decomp:k} for rectangles
    containing at most $2k$ points.  Let $\Family$ be the
    corresponding family of sets for $ P$. We claim that $\Family$
    has the desired property.

    \phantom{}\hfill%
    \includegraphics[page=1]{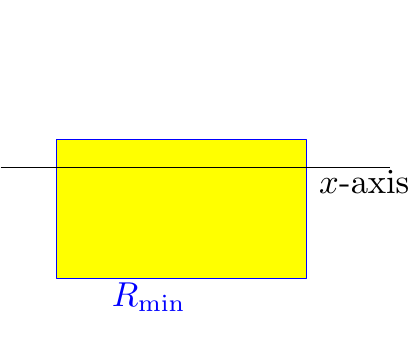}\hfill
    \includegraphics[page=2]{figs/reflect}\hfill
    \includegraphics[page=3]{figs/reflect}%
    \hfill%
    \phantom{}%
    
    Let $\rect = \foldX{\rmin}$ be the folding of $\rmin$, and let
    $\rect' =\refX{\rect}$. Observe that
    $f(\rect) = f(\rect') \leq f(\rmin)$ because of the translation
    invariance of $f$, and monotonicity of $f$.

    If $ |\rect \cap  P| > k$ then one can shrink it so that it
    contains only $k$ points of $ P$, but this would imply that
    $\rmin$ is not the minimum, a contradiction. The case that
    $ |\rect' \cap  P| > k$ leads to a similar contradiction. We
    conclude that $\rect \cup \rect'$ contains at most $2k$ points of
    $ P$. Implying that $\rect = \foldX{\rmin}$ contains at most $2k$
    points of $ P'$.  As such, there is a set $ Q' \in \Family'$
    that contains $\rect \cap  P'$. Now, let $ Q$ be the
    corresponding set in $\Family$ to $ Q'$. Since
    $\rmin \subseteq \rect \cup \rect'$, it follows that
    $\rmin \cap  P \subseteq (\rect \cup \rect') \cap  P \subseteq
     Q$, as desired.
\end{proof}

\begin{lemma}
    \lemlab{family}%
    Given a set $ P$ of $n$ points in the plane, a parameter $k$, and
    a scoring function $f$, let $\rmin$ be the minimum score rectangle
    that contains $k$ points of $ P$. One can compute, in
    $O(n \log n )$ time, a family $\Family$ of
    $O( \tfrac{n}{k} \log \tfrac{n}{k} )$ subsets of $ P$, such that
    (i) each subset of $\Family$ is of size $\leq 12k$, and (ii)
    $\rmin \cap  P \subseteq  Q$, for some $ Q \in \Family$.
\end{lemma}
\begin{proof}
    Find a horizontal line $\Line$ that splits $ P$ evenly, and
    compute the family of \lemref{family:horizontal}. Now recurse on
    the points above $\Line$ and the points below $\Line$. The
    recursion bottoms out when the number of points is $\leq 12k$. The
    correctness is by now standard -- as soon as a recursive call
    picks a line that stabs the optimal rectangle, the family
    generated for this line contains the desired set.  The
    $x$-coordinates need to be pre-sorted just once at the beginning.
\end{proof}

\begin{theorem}
    \thmlab{k:s:line}%
    Let $ P$ be a set of $n$ points in the plane, $k$ be a
    parameter, $f$ be a scoring function for rectangles, and let \Alg
    be an algorithm the computes, in $\Talg(m)$ time, the axis-aligned
    rectangle containing $k$ points in a set of $m$ points that
    minimizes $f$. Then one can compute the rectangle containing $k$
    points of $ P$ that minimizes $f$, in time
    \begin{math}
        O\bigl(n\log n + \Talg(12k) \tfrac{n}{k} \log \tfrac{n}{k}
        \bigr).
    \end{math}
\end{theorem}
\begin{proof}
    Compute the family of sets $\Family$ using \lemref{family}, and
    then apply $\Alg$ to each set in this family.
\end{proof}

Combining \thmref{krect} with \thmref{k:s:line} gives the following.
\begin{corollary}%
    \corlab{krect:sensitive:area}%
    Given a set $ P$ of $n$ points in the plane and an integer $k$,
    one can compute, in $O(nk \log \tfrac{n}{k} \log k )$ time, the
    smallest-area axis-aligned rectangle enclosing $k$ points of
    $ P$.
\end{corollary}

For the case when $ t=n-k$ is very small (i.e., finding the smallest
enclosing axis-aligned rectangle with $ t$ \emph{outliers}), there is
an easy reduction (for both perimeter and area) yielding
$O(n+ \Talg(4 t))$ time~\cite{sk-ekpsa-98,Ata}, by keeping the $ t$
leftmost/rightmost/topmost/bottommost points.  Immediately from
\thmref{krect}, we get $O(n + t^2\log t)$ running time.

\subsection{An approximation algorithm for smallest area}
\seclab{min:area:approx}%

In this subsection, we give an efficient approximation algorithm for
the smallest-area $k$-enclosing rectangle problem.  The smallest
perimeter case is straightforward to approximate, by grid rounding,
but the area case is tougher, again because the optimal rectangle may
be long and thin.

\begin{figure}[h]
    \centerline{\includegraphics[scale=0.999]{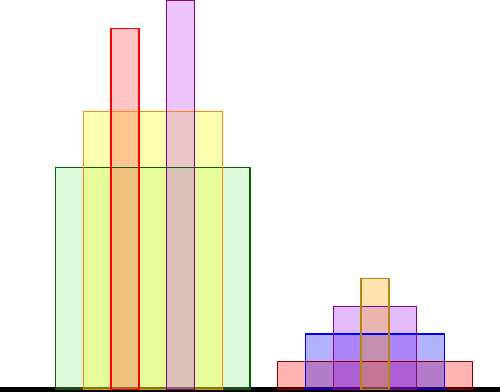}}
\end{figure}
\begin{defn}
    A \emphi{laminar} family of 3-sided rectangles is a collection
    $\RR$ of axis-aligned rectangles with the bottom edges lying on
    the $x$-axis, such that for every pair of rectangles
    $[a,b]\times [0,c]$ and $[a',b']\times [0,c']$ in $\RR$, one of
    the following is true:
    \begin{compactitem}
        \smallskip%
        \item $[a,b]\cap [a',b']=\emptyset$, or %
        \smallskip%
        \item $[a,b]\subseteq [a',b']$ and $c > c'$, or \smallskip%
        \item $[a',b']\subseteq [a,b]$ and $c' > c$.
    \end{compactitem}
\end{defn}

Standard range trees can answer orthogonal range counting queries
(counting the number of points inside rectangular ranges) in
logarithmic time per query (this has been improved to
$O(\sqrt{\log n})$ in the offline setting by Chan and
\Puatracscu~\cite{ChaPat}).  The following lemma shows how to achieve
constant time per query in the offline laminar special case, which
will be useful later in our approximation algorithm.

\begin{lemma}\lemlab{laminar}
    Let $ P$ be a set of $n$ points, and let $\RR$ be a laminar
    family of $O(n)$ 3-sided rectangles in the plane.  Suppose that we
    are given a \emphi{designated point} on the top edge of each
    rectangle in $\RR$, and the $x$- and $y$-coordinates of all the
    designated points and all the points of $ P$ have been
    pre-sorted.  Then we can count, for each rectangle $R\in\RR$, the
    number of points of $ P$ inside the rectangle, in $O(n)$ total
    time.
\end{lemma}
\begin{proof}
    We describe a sweep algorithm, with a horizontal sweep line $\ell$
    moving downward.  Let $X$ denote the set of $x$-coordinates of all
    points of $P$ and all designated points of the rectangles of
    $\RR$.  Consider the union of $R\cap\ell$ over all $R\in\RR$; it
    can be expressed as a union of disjoint intervals.  Let $\II_\ell$
    be the $x$-projection of these disjoint intervals.  Store the
    following collection $\Gamma_\ell$ of disjoint sets in a
    \emph{union-find\/} data structure~\cite{Tar}: for each interval
    $I\in\II_\ell$, define the set $X\cap I$, and for each $a\in X$
    not covered by $\II_\ell$, define the singleton set $\{a\}$.
    Create a linked list $L_\ell$ containing these sets in
    $\Gamma_\ell$ ordered by $x$.  For each set in $\Gamma_\ell$, we
    store a count of the number of points of $P$ below $\ell$ with
    $x$-coordinates inside the set.

    \begin{figure}[h]
        \includegraphics[page=1]{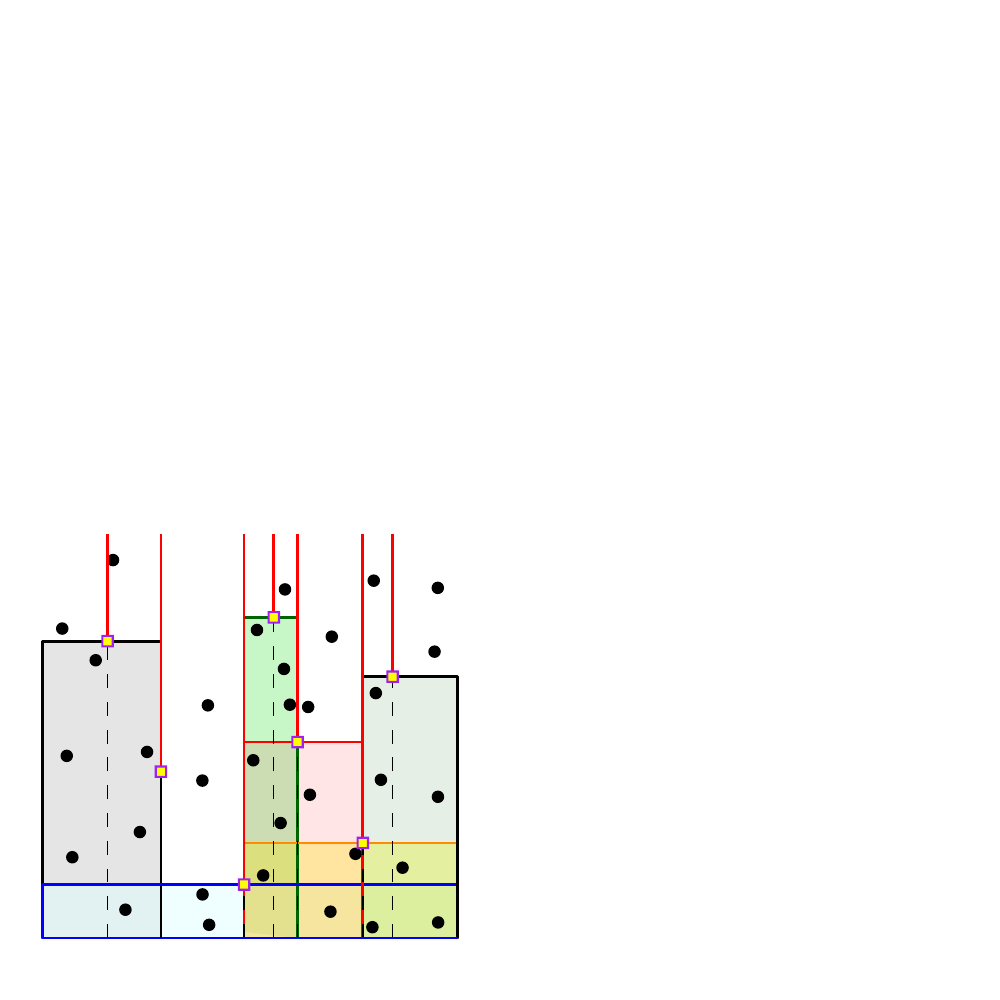}%
        \hfill%
        \includegraphics[page=2]{figs/sweep_down}%
        \hfill%
        \includegraphics[page=3]{figs/sweep_down}%
    \end{figure}
    
    Suppose that the sweep line $\ell$ hits the top edge of a
    rectangle $R$ with $x$-projection $[a,b]$.  By definition of a
    laminar family, any interval in $\II_\ell$ that intersects $[a,b]$
    must be contained in $[a,b]$.  We find the set in $\Gamma_\ell$
    that contains the $x$-coordinate of the designated point of $R$.
    From this set, we walk through the list $L_\ell$ in both
    directions to find all sets contained in $[a,b]$, and replace
    these sets with their union in $\Gamma_\ell$ and $L_\ell$.  The
    count for the new set is the sum of the counts of the old sets;
    this also gives the output count for the rectangle $R$.

    Next, suppose that the sweep line $\ell$ hits a point $p\in P$.
    We find the set in $\Gamma_\ell$ that contains the $x$-coordinate
    of $p$, and decrement its count.  (For example, if the set is a
    singleton, its count changes from 1 to 0.)

    The entire sweep performs $O(n)$ union and find operations.  Gabow
    and Tarjan~\cite{GabTar} gave a linear-time union-find algorithm
    for the special case where the ``union tree'' is known in advance;
    their algorithm is applicable, since the union tree here is just a
    path of the elements ordered by $x$.
\end{proof}

We first solve the approximate decision problem:

\begin{lemma}
    Given a set $ P$ of $n$ points in the plane, a value $\areaVal$,
    and parameters $k$ and $\eps \in (0,1)$, one can either compute
    a $k$-enclosing axis-aligned rectangle $\rect'$ such that
    $\areaX{\rect'} \leq \bigl(1+O(\eps)\bigr) \areaVal$, or conclude
    that the smallest-area $k$-enclosing axis-aligned rectangle has
    area greater than $\areaVal$.  The running time of the algorithm
    is $O\bigl(\eps^{-3}\log \eps^{-1}\cdot n\log n\bigr)$.
\end{lemma}
\begin{proof}
    It is sufficient to solve the problem for the case where the
    rectangle must intersect a horizontal line $\Line$ in
    $O((1/\eps)^3\log(1/\eps)\cdot n)$ time, assuming that the $x$-
    and $y$-coordinates of the given points $ P$ have been
    pre-sorted.  Then standard divide-and-conquer by $y$-coordinates
    gives an $O((1/\eps)^3\log(1/\eps)\cdot n\log n)$-time algorithm
    for the general problem.  Pre-sorting needs to be done only once
    at the beginning.
    
    Without loss of generality, assume that $\Line$ is the $x$-axis.
    Suppose there exists a rectangle $\rect^*$ intersecting $\Line$
    that contains at least $k$ points and has area at most $\areaVal$.
    By symmetry, we may assume that $\rect^*$ has greater area above
    $\Line$ than below, and that the top edge passes through some
    input point $ p=( p_x, p_y)\in  P$.  Then the height $h^*$ of
    $\rect^*$ is between $ p_y$ and $2 p_y$, and the width $w^*$ is
    at most $\areaVal/ p_y$.
    
    Without loss of generality, assume that all $x$-coordinates are in
    $[0,1/3]$.  Define a one-dimensional \emphi{quadtree interval\/}
    (also known as a \emph{dyadic interval\/}) to be an interval of
    the form $[\tfrac{m}{2^i},\tfrac{m+1}{2^i}]$.  It is known that
    every interval of length $w<1/3$ is contained in a quadtree
    interval of length $O(w)$ after shifting the interval by one of
    two possible values $s\in\{0,1/3\}$ (this is a special case of a
    shifting lemma for $d$-dimensional quadtrees~\cite{Bern,Chan98}).
    Thus, suppose that $[p_x-\areaVal/p_y, p_x+\areaVal/p_y]$ is
    contained in the interval $[\tfrac{m}{2^i}+s,\tfrac{m+1}{2^i}+s]$,
    where the length $\tfrac{1}{2^i}$ is equal to the smallest power
    of 2 greater than $c \areaVal/p_y$, for some constant $c$.  (Note
    that $i$ is a nondecreasing function of $p_y$.)  Without loss of
    generality, assume that $1/\eps = 2^E$ for an integer $E$.  Define
    a family of $O(1/\eps^3)$ \emphi{canonical} rectangles of the form
    \begin{equation*}
        [\tfrac{m}{2^i} + \tfrac{j}{2^{i+E}} + s,\ 
        \tfrac{m}{2^i} + \tfrac{j'}{2^{i+E}} + s]
        \ \times\
        [-j''\eps p_y, p_y]        
    \end{equation*}
    over all possible indices $j,j',j''\in \{0,\ldots,1/\eps\}$ such
    that
    $p_x\in [\tfrac{m}{2^i} + \tfrac{j}{2^{i+E}} + s,\ \tfrac{m}{2^i}
    + \tfrac{j'}{2^{i+E}} + s]$.
    
    By rounding, $R^*$ is contained in a canonical rectangle $R'$ with
    height at most $h^*+O(\eps)p_y\le (1+O(\eps))h^*$ and width at
    most
    \begin{equation*}
        w^*+O(\eps)\areaVal/p_y \leq \bigl(1+O(\eps)\bigr)\areaVal/h^*,
    \end{equation*}
    and thus area at most $(1+O(\eps))\areaVal$.  So, it suffices to
    count the number of points inside each canonical rectangle and
    return the smallest area among those rectangles containing at
    least $k$ points.
    
    To speed up range counting, observe that for canonical rectangles
    with the same $j,j',j''$, the same $s\in\{0,1/3\}$, and the same
    value for $(i\bmod E)$, the portion of the rectangles above
    (resp.\ below) $\Line$ forms a laminar family.  This is because:
    (i)~in the $x$-projections, if a pair of intervals intersects, one
    interval must be contained in the other; (ii)~as the height of the
    3-sided rectangle increases, $p_y$ increases, and so $i$ can only
    increase (or stay the same), and so the width of the rectangle can
    only decrease (or stay the same).  Thus, we can apply
    \lemref{laminar} to compute the counts of the points inside each
    rectangle, for all canonical rectangles with a fixed $j,j',j'',s$
    and $(i\bmod E)$, in $O(n)$ time (for each canonical rectangle, we
    can use a point $(p_x,p_y)\in P$ on the top edge, and a
    corresponding point $(p_x,-j''\eps p_y)$ on the bottom edge, as
    the designated points).  The number of choices for $j,j',j'',s$
    and $(i\bmod E)$ is $O((1/\eps)^3\log(1/\eps))$.
\end{proof}

\begin{theorem}
    \thmlab{smallest:area:approx}%
    Given a set $ P$ of $n$ points in the plane, and parameters $k$
    and $\eps \in (0,1)$, one can compute a $k$-enclosing rectangle
    $\rect'$ such that $\areaX{\rect'} \leq (1+\eps) \optY{ P}{k}$,
    where $\optY{ P}{k}$ is the area of the smallest axis-aligned
    rectangle containing $k$ points of $ P$.  The expected running
    time of the algorithm is
    $O\bigl( (1/\eps)^3\log(1/\eps)\cdot n \log n \bigr)$.
\end{theorem}

\begin{proof}
    We can use known techniques for reducing optimization problems to
    decision problems.  We give a self-contained description of one
    approach based on Chan's randomized technique~\cite{Chan99}.

    Let $b$ be a sufficiently large constant.  Divide the plane into
    $b$ columns (vertical slabs) each containing $n/b$ points.
    Similarly divide the plane into $b$ rows (horizontal slabs) each
    containing $n/b$ points.  These steps take linear time by invoking
    a selection algorithm $O(b)$ times.  For each quadruple
    $\QQ=(c,c',r,r')$ where $c$ and $c'$ are columns (with $c$ left of
    $c'$ or $c=c'$) and $r$ and $r'$ are rows (with $r$ below $r'$ or
    $r=r'$), consider the subproblem of finding the smallest-area
    rectangle containing $k$ points of $ P$, subject to the extra
    constraints that the left edge of the rectangle lies in $c$, the
    right edge lies in $c'$, the bottom edge lies in $r$, and the top
    edge lies in $r'$.  To solve this subproblem, it suffices to
    consider the at most $4n/b$ points in
    $ P\cap (c\cup c'\cup r\cup r')$.  To ensure that the extra
    constraints are satisfied, we add $4n/b$ copies of the four
    intersection points formed by the right boundary of $c$, the left
    boundary of $c'$, the top boundary of $r$, and the bottom boundary
    of $r'$; and we add $16n/b$ to $k$. (Straightforward modifications
    can be made in the special case when $c=c'$ or $r=r'$.)  Let
    $ P_\QQ$ be the resulting point set of size at most $20n/b$
    points, and $k_\QQ$ be the resulting value of $k$.  We thus have
    $$\optY{ P}{k} = \min_\QQ \optY{ P_\QQ}{k_\QQ}.$$

    To compute an approximation to the minimum, we consider the at
    most $b^4$ quadruples in \emph{random order} $\QQ_1,\QQ_2,\ldots$
    and keep track of an approximate minimum $\areaVal$ with the
    invariant that
    $\areaVal\le \min\{ \optY{ P_{\QQ_1}}{k_{\QQ_1}},\ldots,
    \optY{ P_{\QQ_{i-1}}}{k_{\QQ_{i-1}}}\}< (1+\eps)\areaVal$ after
    the $(i-1)$\th iteration.  Let $\eps'$ be such that
    $(1+\eps')^2=1+\eps$; note that $\eps'=\Theta(\eps)$.  At the
    $i$\th iteration, we run the approximate decision procedure for
    $ P_{\QQ_i}$ twice, at values $\areaVal$ and
    $\areaVal/(1+\eps')$, which allows us to conclude one of the
    following:
    \begin{compactitem}
        \medskip%
        \item $\optY{ P_{\QQ_i}}{k_{\QQ_i}} \ge \areaVal$.  In this
        case, we can continue to the next iteration and the invariant
        is maintained.

        \smallskip%
        \item
        $\areaVal/(1+\eps')\le \optY{ P_{\QQ_i}}{k_{\QQ_i}} <
        (1+\eps')\areaVal$.  In this case, we reset $\areaVal$ to
        $\areaVal/(1+\eps')$ and the invariant is maintained.

        \smallskip%
        \item $\optY{ P_{\QQ_i}}{k_{\QQ_i}} < \areaVal$.  In this
        case, we recursively compute an approximation $\areaVal_i$ to
        $\optY{ P_{\QQ_i}}{k_{\QQ_i}}$, satisfying
        $\areaVal_i\le \optY{ P_{\QQ_i}}{k_{\QQ_i}}<
        (1+\eps)\areaVal_i$.  We reset $\areaVal$ to $\areaVal_i$ and
        the invariant is maintained.
    \end{compactitem}
    \medskip%
    We land in the third case only if $\optY{ P_{\QQ_i}}{k_{\QQ_i}}$
    is the smallest among the $i$ values
    $\optY{ P_{\QQ_1}}{k_{\QQ_1}},\ldots,$
    $\optY{ P_{\QQ_i}}{k_{\QQ_i}}$, which happens with probability at
    most $1/i$.  Thus, the expected number of recursive calls is
    bounded by the $(b^4)$\th Harmonic number
    $\sum_{i=1}^{b^4}1/i < \ln(b^4) + 1$.  The expected running time
    satisfies the recurrence
$$T(n) \:\le\: (4\ln b + 1)\,T(20n/b) + O\bigl( (1/\eps)^3\log(1/\eps)\cdot n \log n \bigr),$$
which gives $T(n)=O\bigl( (1/\eps)^3\log(1/\eps)\cdot n \log n \bigr)$
when $b=1000$, for example.
\end{proof}

\section{Extensions}

\subsection{3-sided smallest $k$-enclosing rectangle}
\seclab{3:sided}

In this subsection, we give a slightly faster algorithm for the
3-sided variant of the problem, finding the smallest-area/perimeter
rectangle enclosing $k$ points, under the restriction that the
bottom edge lies on the $x$-axis.  The improvement uses the latest
result on the (min,+)-convolution problem, and is interesting in view
of a reduction in \secref{lower:bounds} in the reverse direction,
establishing essentially an equivalence of the 3-sided problem to
(min,+)-convolution.
    
\begin{problem}\problab{convol}
    \textsf{(min,+)-Convolution.}  Given real numbers
    $a_0,\ldots,a_{n-1},b_0,\ldots,b_{n-1}$, compute
    $c_\ell = \min_{i=0}^\ell (a_i+b_{\ell-i})$ for all
    $\ell=0,\ldots,2n-2$.
\end{problem}
  
Let $\Tconvol(n)$ be the time complexity of the (min,+)-convolution
problem.  As observed by Bremner \etal~\cite{Bremner}, the problem can
be reduced to (min,+)-matrix multiplication, and using the current
best result by Williams~\cite{w-fapsp-14} (derandomized by Chan and
Williams~\cite{cw-daovm-16}),
$\Tconvol(n)=O(n^2/2^{\Omega(\sqrt{\log n})})$.  We use
(min,+)-convolution to speed up the preprocessing time of the 1D data
structure from \secref{krect}.
\begin{lemma}
    The preprocessing time in \lemref{1d} can be reduced to
    $O((n/q)\Tconvol(q)+q^3)$.
\end{lemma}
\begin{proof}
    Divide the $n\times n$ matrix $M$ vertically into $n/q$
    submatrices $M_1,\ldots,M_{n/q}$ each of dimension $n\times q$.
    For each submatrix $M_i$, we consider the portions of the
    diagonals $k,\ldots,k+q$ that are within $M_i$ -- each such
    portion will be called a \emphi{chunk}.  We precompute the minimum
    of the entries in each chunk.  For a fixed $i$, this is equivalent
    to computing $\min_{qi < j \le q(i+1)}(p_{j+\ell-1}-p_j)$ for all
    $\ell\in\{k,\ldots,k+q\}$.  Notice that after some massaging
    of the sequence (negating, reversing, and padding), this
    computation can be reduced to (min,+)-convolution over $O(q)$
    elements, and can thus be done in $O(\Tconvol(q))$ time.  The
    total time over all $i$ is $O((n/q)\Tconvol(q))$.

    Recall that in the preprocessing algorithm in \lemref{1d}, we need
    to compute the minimum of each fragment in the $k,\ldots,k+q$
    diagonals.  Each fragment can be decomposed into some number of
    disjoint chunks plus $O(q)$ extra elements.  Over all $O(q)$
    diagonals, there are $O(q^2)$ fragments and $O(n/q\cdot q)=O(n)$
    chunks in total.  Thus, we can compute the minima of all fragments
    in $O(q^2\cdot q + n/q\cdot q)=O(q^3 + n)$ time, after the above
    precomputation of the minima of all chunks.
\end{proof}

\begin{theorem}%
    \thmlab{krect:3sided}%
    Given a set $ P$ of $n$ points in the plane and integer $k$,
    one can compute, in $O(n^2/2^{\Omega(\sqrt{\log n})})$ time, the
    smallest-area/perimeter axis-aligned rectangle enclosing $k$
    points of $ P$, under the restriction that the bottom edge lies
    on the $x$-axis.
\end{theorem}
\begin{proof}
    Divide the plane into $n/q$ horizontal slabs each containing $q$
    points, for some parameter $q$ to be set later.

    Take such a slab $\sigma$.  We solve the subproblem of finding a
    smallest $k$-enclosing axis-aligned rectangle under the
    restriction that the top edge is in $\sigma$ and the bottom edge
    is on the $x$-axis.  To this end, we first delete all points above
    $\sigma$ or below the $x$-axis.  We build the 1D data structure
    ${\cal S}$ in the lemma for the $x$-coordinates of the surviving
    points, where the marked points are the $q$ points in $\sigma$.
    The preprocessing time is $O((n/q)\Tconvol(q) + q^3)$.  Then for
    each point $p\in\sigma$, we can compute a smallest $k$-enclosing
    axis-aligned rectangle where the top edge has $p$'s $y$-coordinate
    and bottom edge is on the $x$-axis, by making a copy of
    ${\cal S}$, deleting all points in $\sigma$ above $p$, and
    querying ${\cal S}$.  The time needed for the $O(q)$ deletions,
    and for copying ${\cal S}$, is $O(q^2)$.  The total time over all
    $p\in\sigma$ is $O(q^3)$.

    We return the minimum (by area or perimeter) of all the rectangles
    found.  The overall running time over all $n/q$ slabs $\sigma$ is
    $$O((n/q) \cdot ((n/q)\Tconvol(q)+q^3)).$$ With
    $\Tconvol(q) = O(q^2/2^{\Omega(\sqrt{\log q})})$, we can set
    $q=n^{1/3}$, for example, and obtain the final time bound
    $O(n^2/2^{\Omega(\sqrt{\log q})})$.
\end{proof}

For $k$-sensitive bounds, we can apply the shallow cutting technique
from \secref{reduction:k:sensitive} (which is easier for 3-sided
rectangles) and obtain an $O(n\log n + nk/2^{\Omega(\sqrt{\log k})})$
time bound.

\subsection{Arbitrarily oriented smallest $k$-enclosing rectangle}
\seclab{arb:oriented:excluding}

We briefly consider the problem of computing a smallest-area/perimeter
arbitrarily oriented rectangle (not necessarily axis-aligned)
enclosing $k$ points.  The optimal rectangle is defined by 5 points,
with one edge containing 2 points $p_1^*$ and $p_2^*$.  Given a fixed
choice of $p_1^*$ and $p_2^*$, we can use a rotation and translation
to make $p_1^*p_2^*$ lie on the $x$-axis and thereby obtain a 3-sided
axis-aligned rectangle problem, which can be solved in
$O(n\log n + nk/2^{\Omega(\sqrt{\log k})})$ time.  Exhaustively trying
all pairs $p_1^*p_2^*$ then gives
$O(n^3\log n + n^2k/2^{\Omega(\sqrt{\log k})})$ total time.

\subsection{Minimum-weight $k$-enclosing rectangle}
\seclab{krect:minwt}

Our $O(n^2\log n)$-time algorithm can be adapted to solve the
following related problem.  (Without the $k$ constraint, the problem
has an $O(n^2)$-time algorithm \cite{bcnp-mwpb-14}.)

\begin{theorem}%
    \thmlab{krect:minwt}%
    Given a set $ P$ of $n$ points in the plane each with a real
    weight, and an integer $k$, one can compute, in $O(n^2\log n)$
    time, the axis-aligned rectangle enclosing $k$ points minimizing
    the total weight of the points inside.
\end{theorem}
\begin{proof}
    We follow the same approach as in \secref{krect}, with the
    following differences in the data structure of \lemref{1d}.  For
    every fragment we maintain the minimum weight solution.  Using
    prefix sums, the entry $M_{i,j}$ in the matrix contains the total
    weight of the elements from $i$ to $j$. As before, we break the
    $q+1$ diagonals of entry into fragments, where each fragment
    summary maintains the minimum weight encountered.

    A deletion of a marked point $ p$ of weight $w$ would result is
    an insertion of a fixup entry, of value $-w$ into a linked list of
    a diagonal where $ p$ appeared as a singleton (when crossing a
    column of $ p$), and a fixup entry of value $+w$ when
    encountering the row column of $ p$. The real value of a fragment
    is the value stored in the fragment plus the total sum of the
    fixups appearing before it in the linked list of its diagonal. As
    such, during query the real value can be computed in $O(q)$ time
    overall, as this list is being scanned.  When we merge two
    adjacent fragments separated by a singleton, we should increase
    the later fragment by the fixup value at the singleton before
    taking the minimum. Clearly, all the operations can be implemented
    in $O(q)$ time.

    Now, we can use the divide-and-conquer algorithm in the proof of
    \thmref{krect} with no change.
\end{proof}

As an application, we can solve the following problem: given $n$
points in the plane each colored red or blue, and an integer $k$,
find an axis-aligned rectangle enclosing exactly $k$ points
minimizing the number of red points inside.  This is a special case of
the problem in the above theorem, where the red points have weight 1
and blue points have weight 0, and can thus be solved in
$O(n^2\log n)$ time.

Similarly, we can solve for other variants of the red/blue problem,
for example, finding a $k$-enclosing rectangle maximizing (or
minimizing) the number of red points, or finding a $k$-enclosing
rectangle with exactly a given number $k_r$ of red points.  (For the
latter, the following observation allows us to reduce the 1D
subproblem to querying for the maximum and minimum: given a set $P$ of
red/blue points in 1D and a value $k$, let $K_r$ denote the set of
all possible values $k_r$ for which there exists an interval
containing $k$ points of $P$ and exactly $k_r$ red points; then
$K_r$ forms a contiguous range of integers, and thus contains all
numbers between $\min(K_r)$ and $\max(K_r)$.)

\subsection{Subset sum for $k$-enclosing rectangle}
\seclab{subset:sum:rect}%

A more challenging variant of the weighted problem is to find a
rectangle enclosing exactly $k$ points with total weight exactly $W$
(similar to subset sum), or more generally, find an axis-aligned
rectangle enclosing exactly $k$ points with total weight closest to
$W$.

We use a different approach, using a 1D data structure that is static
but can ``plan for'' a small number of deletions.

\begin{lemma}\lemlab{subset:sum}
    Given a set $P$ of $n$ points in 1D and integers $k$ and $q$, we
    can build a static data structure, with $O(nq\log n)$
    preprocessing time, that supports the following type of queries in
    $O(q\log n)$ time: for any subset $D\subset P$ of at most $q$
    points and any weight $W$, find an interval containing $k$ points
    of $P-D$ with weight closest to $W$.
\end{lemma}
\begin{proof}
    As in the proof of \lemref{1d}, we sort $P$ and consider the
    (implicit) matrix $M=P-P$.  For each $i\in\{k,\ldots,k+q\}$, we
    store the elements in the $i$\th diagonal in a data structure
    supporting 1D \emph{range predecessor/successor} queries -- i.e.,
    finding the predecessor/successor to any value among the elements
    in any contiguous sublist of a given list -- in $O(\log n)$ time,
    after $O(n\log n)$ preprocessing time.  (See \cite{Zhou} for the
    latest result on the range successor problem; for simplicity, we
    will ignore improvements in the logarithmic factors here.)  The
    total preprocessing time for all $q+1$ diagonals is $O(nq\log n)$.

    To answer a query, we imagine deleting the columns and rows
    associated with the elements in $D$, from the matrix $M$.  We want
    to search for $W$ in the $k$\th diagonal in the modified matrix.
    This diagonal corresponds to $O(q)$ \emphi{fragments} from the
    $k,\ldots,k+q$ diagonals in the original matrix.  Namely, as we
    trace the diagonal in the original matrix from left to right,
    whenever we hit a deleted column, we move to the diagonal one unit
    up, and whenever we hit a deleted row, we move the diagonal one
    unit down.  We can search for $W$ in each fragment by a range
    predecessor/successor query in a diagonal in $O(\log n)$ time.  We
    return the closest point found.  The total query time is
    $O(q\log n)$.
\end{proof}

\begin{theorem}
    \thmlab{subset:sum}%
    Given $n$ points in the plane each with a real weight, and given a
    real number $W$ and an integer $k$, one can compute, in
    $O(n^{5/2}\log n)$ time, an axis-aligned rectangle enclosing
    exactly $k$ points with total weight closest to $W$.
\end{theorem}
\begin{proof}
    Divide the plane into $n/q$ horizontal slabs each containing $q$
    points, for some parameter $q$ to be set later.

    Take a pair of horizontal slabs $\sigma$ and $\tau$.  We solve the
    subproblem of finding a $k$-enclosing axis-aligned rectangle with
    weight closest to $W$, under the restriction that the top edge is
    in $\sigma$ and the bottom edge is in $\tau$.  To this end, we
    first remove all points strictly above $\sigma$ and strictly below
    $\tau$.  We build the 1D data structure in \lemref{subset:sum} for
    (the $x$-coordinates of) the surviving points in $O(nq\log n)$
    time.  Then for each pair of points $p_\sigma\in\sigma$ and
    $p_\tau\in \tau$, we can search for a $k$-enclosing axis-aligned
    rectangle with weight closest to $W$, under the restriction that
    the top edge has $p_\sigma$'s $y$-coordinate and the bottom edge
    has $p_\tau$'s $y$-coordinate, by performing a query to the 1D
    data structure in $O(q\log n)$ time for the subset $D$ of (the
    $x$-coordinates of) the points above $p_\sigma$ in $\sigma$ and
    the points below $p_\tau$ in $\tau$.  The total query time over
    all $O(q^2)$ pairs of points $(p_\sigma,p_\tau)$ is
    $O(q^3\log n)$.

    We return the closest answer found.  The total running time over
    all $O((n/q)^2)$ pairs of slabs $(\sigma,\tau)$ is
    $$O((n/q)^2 \cdot (nq\log n + q^3\log n)).$$ We set $q=\sqrt{n}$.
\end{proof}

We can further improve the running time for small $k$:

\begin{theorem}
    \thmlab{subset:sum2}%
    Given $n$ points in the plane each with a real weight, and given a
    real number $W$ and an integer $k$, one can compute, in
    $O(n^2\sqrt{k}\log k)$ time, an axis-aligned rectangle
    enclosing exactly $k$ points with total weight closest to $W$.
\end{theorem}
\begin{proof}
    We first consider the variant of the problem where the rectangle
    is constrained to intersect a fixed vertical line~$\ell$.  Here,
    we can follow essentially the same algorithm as in
    \thmref{subset:sum} (with minor modifications to the data
    structure in \lemref{subset:sum}), but when considering the
    subproblem for the two slabs $\sigma$ and $\tau$, we can further
    remove more irrelevant points: among the points strictly between
    $\sigma$ and $\tau$, it suffices to keep only the $k$ points
    immediately to the left of $\ell$ and the $k$ points immediately
    to the right of $\ell$.  It may be too costly to compute these
    $O(k)$ points from scratch, but given such points for
    $(\sigma',\tau)$ for the predecessor slab $\sigma'$ of $\sigma$,
    we can generate the new $O(k)$ points for $(\sigma,\tau)$, by
    invoking a selection algorithm on $O(k+q)$ elements in $O(k+q)$
    time.  The number of points in the subproblem is now reduced to
    $O(k)$.  As a consequence, the overall running time becomes
    \begin{equation*}
        O\bigl((n/q)^2 \cdot (kq\log k + q^3\log k)\bigr).
    \end{equation*}
    We set $q=\sqrt{k}$.  This gives an $O(n^2\sqrt{k}\log k)$-time
    algorithm for the restricted problem with the vertical line
    $\ell$.

    We can solve the original problem now via standard
    divide-and-conquer by $x$-coordinates, with running time given by
    the recurrence $T(n,k)=2\,T(n/2,k)+O(n^2\sqrt{k}\log k)$, which
    solves to $T(n,k)=O\bigl(n^2\Bigl.\sqrt{k}\log k\bigr)$.
\end{proof}

As an application, we can solve the following problem:
given $n$ colored points in the plane with $d$ different colors, and
integers $k_1,\ldots,k_d$, with $k_1+\cdots+k_d=k$, find an
axis-aligned rectangle enclosing exactly $k_i$ points of the $i$\th
color.
The problem was proposed by Barba \etal \cite{bdfhm-okeoc-13}, who
gave an $O(n^2k)$-time algorithm.  (It may be viewed as a geometric
variant of the jumbled or histogram indexing problem for
strings~\cite{cl-ci3ac-15}.)  It is a special case of the problem from
\thmref{subset:sum}: we can give points with color $i$ a weight of
$M^i$ for a sufficiently large $M$, e.g., $M=n+1$, and set the target
to $W=\sum_{i=1}^d k_i M^i$.  Since weights require $O(d\log n)$ bits,
each addition has $O(d)$ cost, and so the running time becomes
$O(d n^2\sqrt{k}\log k)$.  The weights can be reduced to $O(\log n)$
bits by randomized hashing (for example, by randomly selecting $M$
from $\{0,\ldots,p-1\}$ and working with numbers modulo $p$ for an
$O(\log n)$-bit prime $p$), since there are only polynomially (i.e.,
$O(n^4)$) many combinatorially different rectangles.  This way, the
running time can be reduced to $O(n^2\sqrt{k}\log k)$ -- this improves
Barba \etal's result.

\subsection{Conditional lower bounds}
\seclab{lower:bounds}

We can prove that the smallest-perimeter $k$-enclosing axis-aligned
rectangle problem do not have truly subquadratic
($i.e., O(n^{2-\delta})$) algorithms, under the conjecture that
(min,+)-convolution does not have a truly subquadratic algorithm.  Our
proof holds for the 3-sided version of the problem, which complements
nicely with our upper bound in \secref{3:sided} using
(min,+)-convolution.

We describe a reduction from the following decision problem, which
Cygan \etal \cite{cmww-pempc-17} showed does not have a truly
subquadratic algorithm under the (min,+)-convolution conjecture.

\begin{problem} \textsf{(min,+)-Convolution Decision.}
    \problab{convbound}%
    Given real numbers $a_0,\ldots,a_{n-1},b_0,\ldots,b_{n-1},$ and
    $c_0,\ldots,c_{n-1}$, decide whether
    \begin{equation*}
        \forall\ell:\ c_\ell\le \min_{i+j=\ell} (a_i+b_j).
    \end{equation*}
\end{problem}

\begin{theorem}
    If there is a $T(n)$-time algorithm for computing the
    smallest-perimeter/area axis-aligned rectangle enclosing $k$
    points for a given set of $n$ points in the plane and a given
    number~$k$ (with or without the constraint that the bottom edge
    lies on the $x$-axis), then there is an $O(T(O(n))$-time algorithm
    for \probref{convbound}.
\end{theorem}
\begin{proof}
    Consider an instance of \probref{convbound}.  Without loss of
    generality, assume $a_i,b_j,c_\ell\in (0,1)$.  We create an
    instance of the minimum-perimeter $k$-enclosing rectangle problem
    with $3n$ points
    \begin{equation*}
        \{(-i-a_i,0)\}_{i=0}^{n-1}\:\cup\:
        \{(j+b_j,0)\}_{j=0}^{n-1}\:\cup\:
        \{(0,n-\ell-c_\ell)\}_{\ell=0}^{n-1},
    \end{equation*}
    plus $M$ extra copies of $(-a_0,0)$ and $M$ extra copies of
    $(b_0,0)$, and $k=n+2+2M$, where $M$ is a sufficiently large
    number, e.g., $M=2n$.

    The optimal perimeter is
    \begin{equation*}
        \min_{i,j,\ell\,:\, (i+1)+(j+1)+(n-\ell)+2M=k}
        2(i+a_i+j+b_j+n-\ell-c_\ell)
        \ =\ %
        \min_{i,j,\ell\,:\, i+j=\ell} 2(n+a_i+b_j-c_\ell),
    \end{equation*}
    which is at least $2n$ iff
    $\min_{i,j:i+j=\ell} (a_i+b_j) \ge c_\ell$ for every $\ell$.
    
    For minimum area, the reduction is similar, except that we replace
    $n-\ell-c_\ell$ with $\tfrac{1}{\ell+c_\ell}$.  The optimal area
    is
    \begin{equation*}
        \min_{i,j,\ell\,:\, (i+1)+(j+1)+(n-\ell)+2M=k}
        \frac{i+a_i+j+b_j}{\ell+c_\ell}
        \ =\ %
        \min_{i,j,\ell\,:\, i+j=\ell}
        \frac{\ell+a_i+b_j}{\ell+c_\ell},
    \end{equation*}
    which is at least 1 iff $\min_{i,j:i+j=\ell} (a_i+b_j) \ge c_\ell$
    for every $\ell$.
\end{proof}

A similar reduction holds for the minimum-weight $k$-enclosing
rectangle problem from \thmref{krect:minwt}:

\begin{theorem}
    If there is a $T(n)$-time algorithm for computing the
    minimum-weight axis-aligned rectangle enclosing $k$ points for a
    given set of $n$ weighted points in the plane and number $k$ (with
    or without the constraint that the bottom edge lies on the
    $x$-axis), then there is an $O(T(O(n))$-time algorithm for
    \probref{convbound}.
\end{theorem}
\begin{proof}
    The reduction is similar.  Assume $a_i,b_j,c_\ell\in (0,1)$.  We
    create an instance of the minimum-weight $k$-enclosing rectangle
    problem with $3n$ weighted points
    \begin{equation*}
        \{(-i,0; a_i-a_{i-1})\}_{i=0}^{n-1}\:\cup\:
        \{(j,0;b_j-b_{j-1})\}_{j=0}^{n-1}\:\cup\: \{(0,n-\ell;
        -c_\ell+c_{\ell+1})\}_{\ell=0}^{n-1},
    \end{equation*}
    plus $M$ extra copies of the points $(-1,0;0)$ and $(1,0;0)$, and
    $k=n+2+2M$, where $M$ is a sufficiently large number, e.g.,
    $M=2n$.  The third coordinate after the semicolon of each point
    denotes its weight.  (And $a_{-1}=b_{-1}=c_{n}=0$.)

    The minimum weight over all rectangles with $k$ points is equal to
    \begin{equation*}
        \min_{i,j,\ell\,:\, (i+1)+(j+1)+(n-\ell)+2M=k}
        (a_i+b_j-c_\ell)
        \ =\ 
        \min_{i,j,\ell\,:\, i+j=\ell} (a_i+b_j-c_\ell),
    \end{equation*}
    which is nonnegative iff
    $\min_{i,j:i+j=\ell} (a_i+b_j) \ge c_\ell$ for every $\ell$.
\end{proof}

A near-quadratic conditional lower bound for the minimum-weight
rectangle problem without the $k$ constraint was given by Backurs
\etal~\cite{BDT} (under a different ``popular'' conjecture about the
complexity of maximum-weight clique).

We can similarly prove that the subset-sum variant of the
$k$-enclosing rectangle problem from \thmref{subset:sum} (or its
3-sided variant) does not have truly subquadratic algorithms, under
the conjecture that the \emph{convolution-3SUM\/} problem (given real
numbers $a_0,\ldots,a_{n-1},b_0,\ldots,b_{n-1},c_0,\ldots,c_{n-1}$,
decide whether $c_\ell=a_i+b_{\ell-i}$ for some $i$ and $\ell$) does
not have a truly subquadratic algorithm (which is known to be true
under the conjecture that 3SUM for integers does not have a truly
subquadratic algorithm~\cite{Pat}).

\newcommand{\etalchar}[1]{$^{#1}$}
\providecommand{\CNFCCCG}{\CNFX{CCCG}} \providecommand{\CNFX}[1]{
   {\em{\textrm{(#1)}}}}

\appendix

\section{Shallow cutting for points and $3$-sided %
   rectangles}
\apndlab{shallow:cutting}

Here, we prove a shallow cutting lemma due to \Jorgensen and Larsen
\cite{jl-rsmt-11} -- we provide the full details for the sake of
completeness.

Let $ P$ be a set of $n$ points in general position (i.e., no two
points share $x$ or $y$ values), let $\Line$ be a horizontal line
below all the points of $ P$, and let $k$ be a parameter.  For
simplicity of exposition we assume $\Line$ is the $x$-axis. We
generate $m = O(n/k)$ subsets $ Q_1,\ldots,  Q_m$ of $ P$, each
of size $O(k)$, such that any axis-parallel rectangle $\rect$ that has
its bottom edge on $\Line$, and contains at most $k$ points of $ P$,
its associated subset of $ P$ is contained in one of these subsets.

To this end, we sweep horizontally upward from $\Line$. At any point
in time, the $x$-range is going to be split into interior disjoint
intervals (the two extreme ones are rays).  For an interval
$I = [x_1,x_2]$, its weight at time $t$ is
$n(I,t) = \cardin{  P(I,t) }\bigr.$, where
$ P(I,t) =  P \cap \pth{I \times [0,t]}$ is the active set of points
associated with $I$ at time $t$. As soon as the weight of some active
interval $J$ becomes $2k$, say at time $t$, we stop the sweep and
split $J$ into two intervals, by picking a value $m \in J$ between the
$k$\th rank and $(k+1)$\th rank $x$-coordinates of the points of
$ P(I,t)$, and breaking $I$ at $m$ into two intervals. See
\figref{sweepy}. Let $I^-, I^+$ be the two intervals adjacent to $I$
just before this split.  We store the point $ p_t = (m,t)$ into a set
of splitting points $\Spl$.  Let
$\rect( p_t) = \pth{I^- \cup I \cup I^+} \times [0,t]$ be the
rectangle associated with this split point, and add the set
$ P(I^- \cup I \cup I^+, t) =  P \cap \rect( p_t)$ to the
collection of subsets being computed.

For every two consecutive intervals $J,K$ in the final partition in
the end of the sweeping, we add the set $ P(I\cup J,+\infty)$ to the
collection. Let $\Family$ be the resulting family of sets.

\begin{figure}[h]
    \begin{tabular}{c|c|c}
      \includegraphics[page=1]{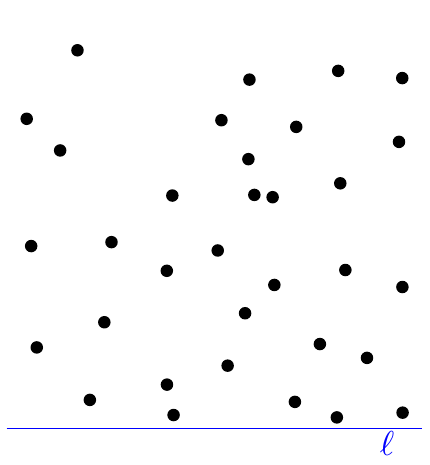}
      &
        \includegraphics[page=2]{figs/shallow}
      &
        \includegraphics[page=3]{figs/shallow}
      \\
      \hline
      \includegraphics[page=4]{figs/shallow}
      &
        \includegraphics[page=5]{figs/shallow}
      &
        \includegraphics[page=6]{figs/shallow}
      \\
      \hline
      \includegraphics[page=7]{figs/shallow}
      &
        \includegraphics[page=8]{figs/shallow}
      &
        \includegraphics[page=9]{figs/shallow}
    \end{tabular}
    \caption{An illustration of \lemref{decomp:k}. The last pane on
       the right shows the point set associated with the last interval
       being split.}
    \figlab{sweepy}
\end{figure}
\paragraph*{Analysis.}
Observe that every time we insert a point $ p$ into $\Spl$, we are
splitting a set of $2k$ points into two sets of $k$ points, Since
there are $n$ points overall, this can happen at most $n/k$
times. This readily implies that $\cardin{\Family} \leq 2 n/k$.

\begin{lemma}
    \lemlab{shallow:cutting}%
    Let $\rect$ be any rectangle having its bottom edge on $\Line$,
    such that $\cardin{\rect \cap  P} < k$. Then, there exists a set
    $ Q \in \Family$, such that $\rect \cap  P \subseteq Q$.
    Furthermore, each set of $\Family$ contains at most $6k$ points.
\end{lemma}
\begin{proof}
    Let $I_\rect$ be the projection of $\rect$ to the $x$-axis. Let
    $I$ be the last active interval that contains $I_\rect$. If $I$ is
    one of the final intervals, then all the points in $I$ strip are
    contained in the set that corresponds to $I$ and its adjacent
    neighbor, and the claim immediately holds.
    
    So, let $t$ be the critical time, where $I$ was split, and let
    $ p$ be the splitting point. There are several possibilities.
    \begin{compactenumI}[label=\Roman*.]
        \item If $ p \notin \rect$, then $\rect(p_t)$ contains
        $\rect$ -- indeed, $\rect(p_t)$ $x$-axis extent contains $I$,
        which contains $I_\rect$. Now,
        $\rect \cap  P \subseteq \rect(p_t) \cap  P$ is in
        $\Family$, and the claim holds.

        \centerline{\includegraphics[page=1]{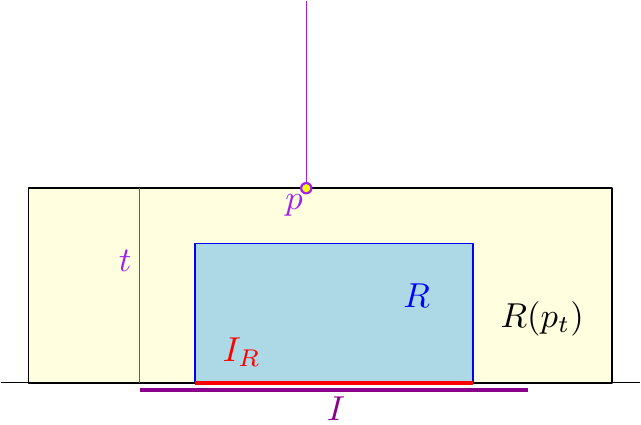}}%

        \item If $ p \in \rect$, then consider the lowest point $\pB$
        (in the $y$-direction) of $\Spl \setminus \brc{ p}$ such that
        its $x$-axis coordinate is in $I_\rect$.

        \begin{compactenumI}[label*=\roman*.]
            \item If $\pB$ does not exist, then $I_\rect$ is contained
            in two consecutive final intervals, and the claim readily
            holds.

            \centerline{\includegraphics[page=2]{figs/split_cutting}}%

            \item If $\pB \in \rect$, then the rectangle
            $\rect' = [x( p),x(\pB)] \times [0, t(\pB)] \subseteq
            \rect$ contains $k$ points of $ P$ (since it is one of the
            sides of the split created by $\pB$ -- which contains
            exactly $k$ points by construction). But that is a
            contradiction to the assumption that
            $| \rect \cap P | < k$.

            \centerline{\includegraphics[page=3]{figs/split_cutting}}%

            \item If $\pB \notin \rect$ then when the upward sweep
            line hits $\pB$, the interval $I_\rect$ is contained in
            two consecutive intervals, where one of them is being
            split by $\pB$. But then, $\rect(\pB)$ is taller than
            $\rect$, and it spans these two intervals. We conclude
            that $\rect \subseteq \rect(\pB)$, which implies the
            claim.

        \end{compactenumI}
    \end{compactenumI}

    As for the size, observe that a set in $\Family$ is the union of
    at most three active sets, and each of these active sets contains
    at most $2k$ points. We conclude that a set of $\Family$ contains
    at most $6k$ points.
\end{proof}

The running time is dominated by $O(n/k)$ queries, each of the
following form: report the lowest $O(k)$ points inside a given
vertical slab.  This is a generalization of range minimum queries, and
known data structures \cite{Brodal} achieve $O(k)$ query time after
$O(n)$ preprocessing time, assuming that the points are given in
$x$-sorted order.  The total construction time is
$O((n/k)\cdot k)=O(n)$.

We thus get the following.

\medskip

{\RestatementOf{\lemref{decomp:k}
      \cite{jl-rsmt-11}}{\LemmaShallowCutting}}

\end{document}